\documentclass[12pt]{article}

\usepackage{amsmath,amsfonts,amssymb}
\usepackage{amsthm}
\usepackage{mathtools}
\usepackage{verbatim}
\usepackage{multicol}
\usepackage{ bbold }
\usepackage{enumitem}
\usepackage{graphicx}
\usepackage[margin=1in]{geometry}
\usepackage{fancyhdr}
\usepackage[capitalize,noabbrev]{cleveref}
\usepackage{indentfirst}
\usepackage{sectsty}

\sectionfont{\large}

\makeatletter
\def\thmhead@plain#1#2#3{%
  \thmname{#1}\thmnumber{\@ifnotempty{#1}{ }\@upn{#2}}%
  \thmnote{ {\the\thm@notefont#3}}}
\let\thmhead\thmhead@plain
\makeatother

\makeatletter
\newcommand{\tpitchfork}{%
  \vbox{
    \baselineskip\z@skip
    \lineskip-.52ex
    \lineskiplimit\maxdimen
    \m@th
    \ialign{##\crcr\hidewidth\smash{$-$}\hidewidth\crcr$\pitchfork$\crcr}
  }%
}
\makeatother

\newtheorem{theorem}{Theorem}
\newtheorem*{theorem*}{Theorem}

\newtheorem{lemma}{Lemma}

\newtheorem{definition}{Definition}
\newtheorem{question}{Question}

\title{Local Quantum Codes from Subdivided Manifolds}
\author{Elia Portnoy \footnote{Mathematics Dept. Massachusetts Institute of Technology,
Cambridge, MA}}
\date{}

\begin{document}
\maketitle

\begin{abstract} For $n \ge 3$, we demonstrate the existence of quantum codes which are local in dimension $n$ with $V$ qubits, distance $V^{\frac{n-1}{n}}$, and dimension $V^{\frac{n-2}{n}}$, up to a $polylog(V)$ factor. The distance is optimal up to the polylog factor. The dimension is also optimal for this distance up to the polylog factor.  The proof combines the existence of asymptotically good quantum codes, a procedure to build a manifold from a code by Freedman-Hastings, and a quantitative embedding theorem by Gromov-Guth.
\end{abstract}

\section*{\centering Introduction}
 Asymptotically good quantum codes are quantum codes that are theoretically very good at correcting errors. Recently, several breakthrough results have given constructions for asymptotically good quantum codes. However, implementing them into a physical architecture presents several difficulties. We can imagine implementing an architecture for a quantum code by placing each qubit at different point in a $3$-dimensional lattice, so that each check of the quantum code only involves the qubits that are placed within a constant radius of some lattice point. We'll refer to quantum codes which admit such an architecture as local, and we'll give a rigorous definition of local later in the introduction. There are some known limitations on local quantum codes, and in particular, asymptotically good quantum codes are far from being local. In this paper we give a procedure for transforming any given quantum code that satisfies certain mild assumptions into a local quantum code. This procedure is related to the procedure of subdividing a triangulated manifold and also involves some probabilistic steps. By applying this procedure to the known asymptotically good quantum codes we show that there are local quantum codes which match the known limitations on them. 

\par We now give some definitions for quantum codes needed to state the results of the paper. A Calderbank-Shor-Steane (CSS) quantum code, denoted here by $\mathcal{C}$, is a 3-term chain complex of $\mathbb{F}_2$ vector spaces with a distinguished set of basis vectors,

$$C_2 \xrightarrow{H_2^{T}} C_1 \xrightarrow{H_1} C_0$$

\noindent where $H_2$ and $H_1$ are matrices defining the boundary maps so that $H_1 H_2^{T} = 0$. The size of $\mathcal{C}$ is $size(\mathcal{C}) = dim(C_1)$. The dimension of $\mathcal{C}$ is,
$$dim(\mathcal{C}) = dim(ker (H_1)/Im (H_2^T)) = dim(ker (H_2)/Im (H_1^T))$$
\noindent Note that $size(\mathcal{C})$ is often referred to as the number qubits, and $dim(\mathcal{C})$ is often referred to as the number of logical qubits. For a vector $v \in C_i$ with $i \in \{0,1,2\}$, we let $|v|$ denote its hamming weight with respect to the basis given by $\mathcal{C}$. The $X$-distance and $Z$-distance of the code are,

$$d_X(\mathcal{C}) = \min(|v|: v \in ker (H_1), v \not\in Im (H_2^T))$$
$$d_Z(\mathcal{C}) = \min(|v|: v \in ker (H_2), v \not\in Im (H_1^T))$$

\noindent The distance of $\mathcal{C}$ is $d(\mathcal{C}) = \min(d_X, d_Z)$. We say that $\mathcal{C}$ is low density parity check (LDPC) if $H_2$ and $H_1$ have at most $A$ non-zero entries in each row and column, where $A>0$ is some universal constant. From now on, we'll refer to LDPC CSS quantum codes simply as codes.

\par The desirable property of codes for good error correction is large dimension and large distance. In the recent decades there has been significant progress in constructions for codes with good error correction properties. We'll just mention a few here and refer to the literature to explain the terminology used for these constructions. In 2009 Z\'emor \cite{Z} asked if there where any codes $\mathcal{C}$ such that $dim(\mathcal{C}) dist(\mathcal{C})^2 > size(\mathcal{C})^{1+\epsilon}$ for some small $\epsilon>0$. In 2013 this was answered positiviely by Guth and Lubotzky by constructing codes from arithmetic hyperbolic 4-manifolds. In 2020 Hastings, Haah and O'Donnell \cite{HHO} constructed codes with distance at least $size(\mathcal{C})^{3/5} polylog^{-1}(\mathcal{C})$ and dimension like $size(\mathcal{C})^{3/5}$, using a twisted fiber bundle construction. In 2020 Breuckmann and Eberhardt \cite{BE} constructed codes with distance at least $size(\mathcal{C})^{3/5}$ and dimension like $size(\mathcal{C})^{4/5}$, using a balanced product of a classical code from a Ramanujan graph and a Sipser-Speilman code. In 2022 Panteleev and Kalachev \cite{PK} found codes with both distance and dimension like $size(\mathcal{C})$. Their construction was an expander lifted product and crucially used a two-sided robustly testable product code. Later that year Leverrier and Z\'emor \cite{LZ} constructed codes with both distance and dimension like $size(\mathcal{C})$, using a two-sided Cayley complex with a two-sided robustly testable product code.

\begin{theorem} [\cite{PK}, \cite{LZ}] For any $V>1$, there exists a LDPC code $\mathcal{C}$ with $size(\mathcal{C}) > V$, $d(\mathcal{C}) > \varepsilon\, size(\mathcal{C})$ and $dim(\mathcal{C}) > \varepsilon\, size(\mathcal{C})$, where $\varepsilon>0$ is some universal constant.
\end{theorem}

We'll refer to codes with  with both distance and dimension like $size(\mathcal{C})$, as in the above theorem, as asymptotically good quantum codes. Next, we give a formal definition of what it means for a code to be local. 

\begin{definition} For an LDPC code $\mathcal{C}$ let $\mathcal{B}$ be the distinguished set of basis vectors for $C_1$. We say that two distinguished basis vectors $e_i, e_j \in \mathcal{B}$, are part of a check if there is a row in either $H_1$ or $H_2$ for which the $i$ and $j$ entries are non-zero. For a dimension $n>0$ and some number $V>0$, define the following $n$-dimensional cube of lattice points,

$$Q^n(V)  = \{(x_1, x_2 \ldots x_n) \in \mathbb{Z}^n: \sum_{i=1}^n |x_i| < AV^{1/n}\}$$

\noindent where $A>0$ is some large constant that only depends on $n$. For two lattice points $x, y \in \mathbb{Z}^n$ define $dist(x, y) = \sum_{i=1}^n |x_i - y_i|$. We say that $\mathcal{C}$ is \textbf{local in dimension $n$} (or just \textbf{local} if $n$ is implicit) if there is an injective map,

$$f: \mathcal{B} \to Q^n(size(\mathcal{C}))$$

\noindent so that every pair of basis vectors $e_i, e_j \in \mathcal{B}$ that are part of a check satisfy,
$$dist(f(e_i), f(e_j)) \le A$$
\end{definition}
In \cite{B}, \cite{BT} and \cite{BPT} it was shown that the distance of a code which is local in dimension $n$ has an upper bounded of roughly $size(\mathcal{C})^{(n-1)/n}$. So, in particular, asymptotically good quantum codes are not local. To give an example of a local code we can consider a toric code of size about $V$, dimension 2, and distance about $V^{1/2}$. This code can be defined as follows. Consider a 2-dimensionsal torus as the square $[0,V^{1/2}] \times [0, V^{1/2}]$ with opposite sides identified. Place a square grid of side-length 1 on this set. The chain complex of the toric code is defined by taking the basis for $C_0$ to be the vertices, basis for $C_1$ to be the edges, and basis for $C_2$ to be the squares of the grid. $H_1$ is the boundary operator on edges and $H_2$ to be the coboundary operator on the edges. We can see that this code is local in dimension 2 by ``folding" the torus into $\mathbb{R}^2$ as follows. Fix coordinates $(x,y) \in [0,V^{1/2}] \times [0, V^{1/2}]$ on the torus and define the folding map into $\mathbb{R}^2$ as,

$$F(x,y) = (|x-\frac{V^{1/2}}{2}|, |y-\frac{V^{1/2}}{2}|)$$

\noindent It is not too hard to define an injection $f: \mathcal{B} \to Q^2(V)$ using this $F$. By the upper bound mentioned above, this toric code has the best distance we can hope for, for a local code in dimension 2. The toric code is also local in any dimension $n \ge 3$, but for these $n$ its distance falls short of the known upper bound. This example is also instructive for us because it suggests a way of showing that a code is local. Namely, first find a geometrical object associated to the code, and then find a way to ``fold" the geometrical object into $\mathbb{R}^n$.

\par Now we fix some notation used to state the main result of the paper. We write $x \gtrsim y$ to mean $x \ge Ay$, for some constant $A>0$ that only depends on dimensional constants like $n$. We write $x \gtrsim_V y$ to mean $x \ge y\, log^{-A}(V)$, for some constant $A>0$ that only depends on dimensional constants like $n$. Finally,  
$x \approx y$ will mean that both $x \gtrsim y$ and $y \gtrsim x$, and similarly
$x \approx_V y$ will mean that both $x \gtrsim_V y$ and $y \gtrsim_V x$.

\par The main result of the paper is the following.

\begin{theorem} \label{Pthm} Fix any dimension $n \ge 3$, and suppose $\mathcal{C'}$ is an LDPC code that admits a sparse lift (see \cite{FH} for a definition of sparse lift). There exists a code $\mathcal{C}$ which is local in dimension $n$ with,
$$size(\mathcal{C}) \lesssim_V size(\mathcal{C'})^{\frac{n}{n-2}}$$
$$dim(\mathcal{C}) = dim(\mathcal{C'})$$
$$d(\mathcal{C}) \gtrsim_V d(\mathcal{C'})size(\mathcal{C'})^{\frac{1}{n-2}}$$
\end{theorem}

\noindent Applying \cref{Pthm} with $\mathcal{C}'$ being an asymptotically good code, we will show the following.

\begin{theorem} \label{Cthm} For any dimension $n \ge 3$ and $V \gtrsim 1$, there exists a code $\mathcal{C}$ which is local in dimension $n$ with,
$$size(\mathcal{C}) \approx V$$
$$dim(\mathcal{C}) \approx_V V^{(n-2)/n}$$ 
$$d(\mathcal{C}) \approx_V V^{(n-1)/n}$$
\end{theorem}

\par We now sketch the proof of \cref{Pthm} focusing on the distance estimate. The first step is to transform the problem about constructing codes into an equivalent problem about triangulated manifolds. We say a triangulated manifold is local in dimension $n$ if there is a coarse map from it to $\mathbb{R}^n$. A coarse map roughly means that each simplex in the manifold is mapped into a ball of small radius and the pre-image of a small ball intersects only a small number of simplices in the manifold. We will show that one can construct a local code from a local manifold, so that the distance of the code corresponds to some  systoles of the manifold. The $k$-dimensional systole of a manifold is the smallest number of simplices needed to represent a $k$-cycle that this not a boundary in the manifold. 

\par In \cite{FH}, Freedman-Hastings give a procedure for constructing a triangulated manifold from a LDPC code, with some mild assimptions, so that the distance of the code corresponds to some systoles of the manifold. We apply this procedure to the given code $\mathcal{C'}$ to get a manifold $M$. Our goal will be to change the geometry of $M$ until it is local, and to then transform the resulting manifold back into a local code. We begin by using a construction called a nerve map to create a coarse map from our manifold to a 2-dimensional simplicial complex. In \cite{GG}, Gromov-Guth show that if $X$ is a simplicial complex with only a few simplices adjacent to each vertex, then there is a map with small pre-images from $X$ to a ball in $\mathbb{R}^n$ of a certain radius. They use the probabilistic method to construct these maps, and the estimate on the radius is almost sharp for expander-like complexes. We use this map from \cite{GG} to create a map from our 2-dimensional simplicial complex to a ball in $\mathbb{R}^n$. Composing this with the nerve map, we get a map from $M$ to $\mathbb{R}^n$.

\par Although the map in \cite{GG} has small pre-images, it stretches each simplex by a large factor so its not actually coarse. In order to make a coarse map we need to subdivide $M$ so that each new simplex gets mapped into a small ball. It is not so clear to us what this subdivision corresponds to on the associated code, which is why we rely on the language of manifolds. After this subdivision, the resulting manifold is local and we can transform it back into a local code $\mathcal{C}$. A important observation is that the subdivision decreases the systoles of our manifold. Because of this decrease, $d(\mathcal{C})$ is much smaller than $d(\mathcal{C'})$. However, we observe that the map from \cite{GG} can be made to stretch each simplex in a uniform way, which allows us to give some lower bound on the systole of the resulting subdivided manifold. This lower bound in turn allows us to estimate $d(\mathcal{C})$. This finishes the sketch of \cref{Pthm}.
\\
\par This paper is organized as follows. Section 1 contains the results from the literature used to associate a manifold to a code, and to associate a code to a manifold. Section 2 has the definition of a local manifold and preliminaries for coarse maps. Section 3 introduces nerve maps and establishes their properties. Section 4 contains the proof of \cref{Pthm} and \cref{Cthm}. Appendix 1 contains a proof of a slightly modified result from \cite{GG} used in the proof of \cref{Pthm}. Appendix 2 proves some upper bounds for the distance and dimension of local codes. The proofs of these upper bounds are a bit shorter but structurally the same as those already known. They can be found in \cite{FHK}, \cite{B}, \cite{BT}, and \cite{BPT}. We end the introduction with a few open questions.
 
\begin{question} Can we remove the $polylog(V)$ factor from \cref{Pthm}? A possible approach to this question would be to remove the polylog factor from the thick embedding theorem of Gromov-Guth \cite{GG}. 
 \end{question}
 
\begin{question} Is there a way of transforming any LDPC code into a local code without reference to a manifold? This might help get rid of the (probably unnecessary) assumption that the code admit a sparse lift. Also the local code in \cref{Pthm} is obtained using a probabilistic construction. Is there a deterministic procedure for transforming any LDPC code into a local code? 
 \end{question}

\begin{question} Suppose the $\mathcal{C}$ is the local code constructed from $\mathcal{C'}$ in \cref{Pthm}. If there is an efficient decoding algorithm for $\mathcal{C'}$, can we find an efficient decoding algorithm for $\mathcal{C}$ as well?
 \end{question}

\section*{\centering Acknowledgments}
The author would like to thank Larry Guth for pointing out the connection between local codes and coarse maps from manifolds, a proof found in the appendix, and several helpful discussions. The author would also like to thank Matt Hastings for several clarifying comments and questions, some of which appear at the end of the introduction.

\section*{\centering 1. Going between codes and manifolds}

The goal of this section is to introduce some properties of triangulated manifolds and to describe the correspondence between codes and triangulated manifolds. 
\\ \par In this paper, a triangulated $d$-manifold is a simplicial complex, such that the union of the simplices adjacent to any vertex form a set homeomorphic to $\mathbb{R}^d$. Given a triangulated $d$-manifold $M$, let $C_k(M)$ be the space of simplicial $k$-chains with coefficients in $\mathbb{F}_2$ and let $C^k(M)$ be the space of simplicial $k$-cochains with coefficients in $\mathbb{F}_2$. Boundary maps are denoted by $\partial_k: C_k(M) \to C_{k-1}(M)$ and coboundary maps are denoted by $\delta^k: C^k(M) \to C^{k+1}(M)$ respectively. The cycles are $Z_k(M) = ker (\partial_k)$ and the cocycles are $Z^k(M) = ker (\delta_k)$. The boundaries are $B_k(M) = Im(\partial_{k+1})$ and the coboundaries are $B^k(M) = Im(\delta_{k-1})$. Note that $\delta^{k} = \partial_{k+1}^{T}$, when viewed as linear maps. We also let $M(k)$ stand for the set of $k$-simplices in $M$, and let $vol(M) = |M(d)|$. 

\begin{definition} We define the \textbf{$k$-systole} and \textbf{$k$-cosystole} of a $M$, respectively, as

$$sys_k(M) = \inf_{z \in Z_k(M)/B_k(M)} vol(z)$$ 
$$sys^k(M) = \inf_{z \in Z^k(M)/B^k(M)} vol(z)$$ 

\noindent where $vol(z)$ is the number of $k$-simplices in the support of $z$. 
If $H_k(M) = 0$ we set $sys_k(M) = 0$, and if $H^k(M) = 0$ we set $sys^k(M) = 0$ as a convention.
\end{definition}

To associate a code to a manifold we will use the following result discussed in Guth-Lubotzky's introduction.

\begin{theorem} \cite{GL} \label{MClemma} Given a $d$-manifold $M$ and a dimension $0<k<d$ we may associate a code to $M$ as follows. For any $0 \le i \le d$, the set $M(i)$ functions as a basis for both $C_i(M)$ and $C^i(M)$. Thus, we can identify $C_k(M)$ and $C^k(M)$ as vector spaces with the $k$-simplices as a basis. The associated code $\mathcal{C}$ is 

$$C^{k+1}(M) \xrightarrow{(\delta^k)^T} C^k(M) \cong C_k(M) \xrightarrow{\partial_k} C_{k-1}(M)$$  

\noindent We have that, 
$$size(\mathcal{C}) = |M(k)|$$
$$dim(\mathcal{C}) = dim(H_k(M)) = dim(H^k(M))$$
$$d_X(\mathcal{C}) = sys_k(M)$$ 
$$d_Z(\mathcal{C}) = sys^k(M)$$
\end{theorem}

To state the next theorem will use the following property of a simplicial complex.

\begin{definition}The degree of a simplicial complex $X$ is the maximum number of simplices adjacent to a vertex in $X$. We say that $X$ has \textbf{bounded degree} if the degree of $X$ is $\lesssim 1$.\end{definition} 

The following result of Freedman-Hastings shows how to associate a triangulated manifold to a code.

\begin{theorem} \cite{FH} \label{CMlemma} Let $\mathcal{C}$ be a LDPC code which admits sparse lift as defined in \cite{FH}. Let $V = size(\mathcal{C})$ and pick dimensions $d \ge 11$ and $4 \le k \le d-4$. Then there is an associated $d$-dimensional triangulated manifold $M$ with bounded degree such that,
$$vol(M) \approx V$$
$$dim(H_k(M)) = dim(\mathcal{C})$$
$$sys_k(M) \gtrsim d_X(\mathcal{C})$$ 
$$sys_{d-k}(M) \gtrsim d_Z(\mathcal{C})$$
\end{theorem}

This is almost Theorem 1.2.1 in Freedman-Hastings. Their theorem states the additional requirement that $M$ is simply connected and they have $vol(M) \approx_V V$. If we drop the requirement of $M$ being simply connected then their method of proof gives $vol(M) \approx V$. Theorem 1.2.1 also does not state the relationship of the code to $dim(H_k(M))$ as we have done above. However, this bound can be deduced by tracing through their construction and using their Lemma 1.6.2.
\section*{\centering 2. Coarse maps and nerve maps}

In this section we describe a property of a manifold which guarantees this its associated code is local.

\begin{definition} For a set $S$ in a simplicial complex, let $NS$ denote the simplices that intersect $S$. Let $X, Y$ be simplicial complexes. For $A>0$, we call a map $F: X \to Y$ \textbf{$A$-coarse} if for every simplex $\sigma \in X$ and every simplex $\sigma' \in Y$ we have,

$$|NF(\sigma)| \le A$$
$$|NF^{-1}(\sigma')| \le A$$

\noindent We will simply say that $F$ is a \textbf{coarse map} if it is $A$-coarse for some $A \lesssim 1$. 
\end{definition}
Coarse maps obey the following composition rule.
\begin{lemma} \label{complemma} For simplicial complexes $X,Y$ and $Z$ let $F: X \to Y$ be a $A_1$-coarse map and let $G: Y \to Z$ be a $A_2$-coarse map.
Then, $G \circ F$ is a $A_1A_2$-coarse map.
\end{lemma}

\begin{proof} By definition, for any simplex $\sigma \in X$, $|NF(\sigma)|\le A_1$, and $|NG(NF(\sigma))| \le A_2 |NF(\sigma)|$. Since $N(G \circ F)(\sigma) \subset NG(NF(\sigma))$ we have,

$$|N(G \circ F)(\sigma)|  \le A_1A_2$$

We also have, for any simplex $\sigma' \in Z$, $|NG^{-1}(\sigma')|\le A_2$, and $|NF^{-1}(NG^{-1}(\sigma'))| \le A_1 |NG^{-1}(\sigma')|$. Since $N(G \circ F)^{-1}(\sigma') \subset NF^{-1}(NG^{-1}(\sigma'))$ we have,

$$|N(G \circ F)^{-1}(\sigma')|  \le A_1A_2$$
\end{proof}

We now state the definition of a local complex and show how it corresponds to a local code.

\begin{definition} Give $\mathbb{R}^n$ a triangulation by standard unit $n$-simplices. We say that a $d$-dimensional simplicial complex $X$ is \textbf{local in dimension $n$} (of just \textbf{local} if $n$ is implicit) if there is a coarse embedding of $X$ into a ball of volume $vol(X)$ in $\mathbb{R}^n$ with this triangulation.
\end{definition}

The following lemma explains why we use ``local" to describe both codes and manifolds.

\begin{lemma} \label{locallemma} If $M$ is a local $d$-dimensioanl manifold and $\mathcal{C}$ is its associated code from \cref{MClemma}, then $\mathcal{C}$ is local.
\end{lemma} 

\begin{proof} Let $F:M \to \mathbb{R}^n$ be the given $A$-coarse map with $A \lesssim 1$ and $k$ a dimension between $1$ and $d$. For a simplex $\sigma \in M$ let $c(\sigma)$ denote its barycenter. Let $\varepsilon \lesssim 1$ be some constant small enough so that any unit ball in $\mathbb{R}^n$ contains at least $100A$ points of $\varepsilon \mathbb{Z}^n$. Then we can find an injective map,

$$\tilde f: M(k) \to \varepsilon \mathbb{Z}^n$$

\noindent so that for each $\sigma \in M(k)$, $dist(\tilde f(\sigma), F(c(\sigma))) \le 1$ and $|\tilde f(\sigma)| \lesssim V^{1/n}$. Define $f: M(k) \to \mathbb{Z}^n$ by $f(\sigma) = \varepsilon^{-1} \tilde f(\sigma)$. Observe that if $\partial_k(\sigma) \cap \partial_k(\sigma') \neq 0$ or $\delta(\sigma) \cap \delta(\sigma') \neq 0$, then we have,

$$dist(F(c(\sigma)), F(c(\sigma'))) \lesssim \varepsilon^{-1}A$$

Let $\mathcal{C}$ be the code associated to $M$ from the previous proposition. The basis of the middle vector-space of $\mathcal{C}$, $C_1$, corresponds to the $k$-simplices in $M$. Two basis elements are part of one check, if the boundaries or coboundaries of the corresponding simplices intersect. The above observation shows that if two bits in $\mathcal{C}$ are part of one check then they get mapped under $f$ to two points in $\mathbb{Z}^n$ that are $\lesssim 1$ apart. Thus, $\mathcal{C}$ is local.
\end{proof}

Later, we will use the following technical observation.

\begin{lemma}\label{enlargelemma} Suppose $M$ is a $d$-manifold with bounded geometry, $B$ is a ball of volume $V' \gtrsim 1$ in $\mathbb{R}^n$, and $F: M \to B$ is a coarse map. Then there is a local manifold $M''$ with $vol(M'') \approx V'$, $H_i(M'') = H_i(M)$, and $sys_i(M'') \approx sys_i(M)$ for $0<i<d-1$. 
\end{lemma}

\begin{proof} Notice, that from the definition of a coarse map we have $V' \gtrsim vol(M)$, but $V'$ can be much bigger than $vol(M)$. We'll expand $M$ to fill the remaining space in $B$ to make it local. Begin by defining a $d$-sphere $S$ which looks 1-dimensional and has volume roughly $V'$. More specifically, $S$ is the boundary of the $10$-neighborhood of the line segment $[0,V'] \times \{0\}^d \subset \mathbb{R}^{d+1}$. Topologically, $M''$ will be the connect sum of $M$ with $S$. As a simplicial complex we can make $M''$ as follows. Let $\sigma_M$ be a simplex in $M$ and $\sigma_S$ be a simplex in $S$. Let $T$ be $S^{d-1} \times [0,1]$. Then we can glue $T$ to $M$ by identifying $S^{d-1} \times \{0\}$ with $\partial \sigma_M$ and we can glue $T$ to $S$ by identifying $S^{d-1} \times \{1\}$ with $\partial \sigma_S$. We can then triangulate $T$ to be compatible with the identifications and so that $vol(T) \lesssim 1$. $M''$ is the resulting manifold, with $vol(M'') \approx vol(M)+vol(S) \approx V'$. 

\par There is a coarse map $F_1: S \to B$, which is given by collapsing $S$ to a 1-dimensional segment of length $V'$ and winding that segment inside $B$. Combining $F_1$ and $F$ we can make coarse map $F_2: M'' \to B$. Thus, $M''$ is a local manifold. Since $M'$ is the connect sum of $M$ and $S$, the Mayer-Vietoris sequence shows that
$H_i(M'') = H_i(M) \oplus H_i(S) = H_i(M)$, for $0<i<d-1$. To estimate the systole of $M'$, consider an $i$-cycle in $M'$ which we'll denote $z$. Let $z_M = z \cap (M - \sigma_M) \in C_i(M'')$. Then $\partial z_M$ is a $(i-1)$-chain which we can fill in $\partial \sigma_M$ with $\lesssim 1$ simplices. Let $z_M'$ be $z_M$ plus this filling, and notice that it is an $i$-cycle in $M$. Because $S$ has no homology, except in dimensions $0$ and $d$, $z_M'$ is  homologous to $z_M$. Since $vol(z_M') \lesssim vol(z)$ for any $i$-cycle $z$, we can conclude that $sys_i(M'') \approx sys_i(M)$.
\end{proof}
  
\section*{\centering 3. Nerves and nerve maps} 
  
In order to construct coarse maps from a manifolds to a low dimensional Euclidian spaces we will make use of a construction called a nerve. The author learned about nerves from \cite{G} pg 259. We begin by defining a metric on our simlicial complexes. 

\begin{definition} From now on any simplicial complex $X$ will be assumed to have a piecewise linear metric which makes each simplex into a standard simplex of side-length 1. Note that subdividing $X$ changes it metric, while leaving its topology the same. 
\par For two points $p, q \in X$, the distance between $p$ and $q$ in $X$ will be denoted by $dist_X(p,q)$. If $X,Y$ are two simplicial complexes, we say a map $F: X \to Y$ is \textbf{$L$-Lipschitz} if for all $p,q \in X$ we have,
$$dist_Y(F(p), F(q)) \le L\, dist_X(p,q)$$
$F$ is $L$-bilipschitz if we also have, $$dist_X(p,q) \le L\, dist_Y(F(p), F(q))$$

\end{definition}

Next, we define a special class of covers on our simplicial complexes. An open cover of $X$ is a collection of open subsets of $X$, so that every point of $X$ is inside at least one of these subsets.

\begin{definition}
Let $\mathcal{U} = \{U_i\}$ be an open cover of $X$. The \textbf{multiplicity} of $\mathcal{U}$ is the largest number of sets in the cover that have a mutual point of intersection. For any $U_i \in \mathcal{U}$, denote, $$U^r_i = \{p \in U_i: dist(p, \partial U_i) > r\}$$
The cover is called good if the following conditions hold for some positive universal constants $A_1, A_2$ and $D$.

\begin{enumerate}
\item The cover has multiplicity at most $A_1$
\item Each $U_i \in \mathcal{U}$ is contained in a ball of radius $A_2$
\item $X$ is covered by $\{U^{D}_i\}$
\end{enumerate}

\end{definition}

\noindent The next lemma shows that if $\mathcal{U}$ is a good cover we can find a partition of unity corresponding to it with small Lipschitz constants.

\begin{lemma} Let $M$ be a triangulated manifold. We can find a set of functions $\{\psi_i: M \to [0,1]\}_{U_i \in \mathcal{U}}$, with the following properties. 
\begin{enumerate}
\item Each $\psi_i$ is $0$ outside $U_i$.
\item Each $\psi_i$ is $\lesssim 1$-Lipschtiz.
\item For any $p \in M$ we have, $\sum_{U_i \in \mathcal{U}} \psi_i(p) = 1$
\end{enumerate}
\end{lemma} 

\begin{proof} For each $i$, define $\tilde \psi_i: M \to \mathbb{R}$ to be $\tilde \psi_i(p) = dist(p, \partial U_i)$ on $U_i$ and 0 outside $U_i$. Then let,

$$\psi_i(p) = \frac{\tilde \psi_i(p)}{\sum_{U_k \in \mathcal{U}} \tilde \psi_k(p)}$$ 

\noindent These functions automatically satisfy the first and third condition. To check the second condition, we need to show that the norm of the gradient of $\psi_i$ is bounded by a universal constant. This follows from the following properties. Each $\tilde \psi_i(p)$ has Lipschitz constant 1. Since each piece is contained inside a ball of radius $A_2$ we have $|\tilde \psi_i(p)| \lesssim 1$. Since the cover has bounded multiplicity and since $\{U_i^{D}\}$ cover $M$, the denominator is bounded from below and above by a universal constant. 
\end{proof}

Suppose $\mathcal{U}$ is a good cover of $M$ with $(m+1)$ sets. Let $\Delta^m$ denote the $m$-dimensional simplex, which we can view as a subset of $\mathbb{R}^{m+1}$ by defining $\Delta^m = \{(x_1, x_2 \ldots x_{m+1}): \sum_{i=1}^{m+1} x_i =1, x_i \ge 0\}$. Then define a \textbf{nerve map} $\rho: M \to \Delta^m$ as follows. For $\{\psi_i\}$ as in the lemma above, we set,

$$\rho(p) = (\psi_1(p), \psi_2(p) \ldots \psi_{m+1}(p))$$
\noindent 
Notice that by the lemma above, each $\psi_i(p)$ has Lipschitz constant $\lesssim 1$, so the Lipschitz constant of $\rho$ is $\lesssim 1$ as well. The \textbf{nerve} of $M$ for the cover $\mathcal{U}$ is the smallest simplicial complex in $\Delta^m$ that contains $\rho(M)$. Notice that if the multiplicity of the cover is $A$, then $\rho(M)$ lies in the $(A-1)$ skeleton of $\Delta^m$, and so the nerve has dimension less than $A$. If each set of the cover intersects $\lesssim 1$ other sets in the cover, then we also see that the nerve has bounded degree.

\section*{\centering 4. Local manifolds from LDPC codes with lifts}

In this section we describe how to construct a local manifold from an LDPC code using the Freedman-Hastings construction. We then use this constructed manifold to prove \cref{Pthm} and \cref{Cthm}.

\begin{theorem} \label{mfldlemma} Let $\mathcal{C'}$ be an LDPC code that admits a sparse lift with $size(\mathcal{C'}) = V \gtrsim 1$. Fix dimensions $d \ge 11, 4 \le k \le d/2-1$ and $n \ge 3$. Then there is a triangulated $d$-manifold $M$ which is local and satisfies,

$$vol(M) \lesssim_V V^{\frac{n}{n-2}}$$
$$dim(H_k(M)) = dim(\mathcal{C'})$$
$$\min(sys_k(M), sys_{d-k}(M)) \gtrsim_V d(\mathcal{C'})V^{\frac{1}{n-2}}$$ 
\end{theorem}

\begin{proof} 
Let $M$ be the triangulated $d$-manifold constructed from $\mathcal{C'}$ using \cref{CMlemma}. The code $\mathcal{C'}$ consists of a 3-term chain complex with chosen basis representatives,

$$C_2 \to C_1 \to C_0$$

\noindent \cref{CMlemma} says that $vol(M) \approx V$, $dim(H_k(M)) = dim(\mathcal{C'})$, $sys_{k}(M) \gtrsim d(\mathcal{C'})$ and $sys_{d-k}(M) \gtrsim d(\mathcal{C'})$. $M$ is constructed as the double of some manifold, which will be denote by $M_h$. Let be the double cover be denoted as $F_h: M \to M_h$. The key fact about $M_h$ which we will need is the following.

\begin{lemma} \label{multlemma} $M_h$ has a good cover with multiplicity 3. 
\end{lemma}

\begin{proof}[Proof of \cref{multlemma}] For the proof we trace through the construction of $M_h$ from \cite{FH}. $M_h$ is constructed like a handle-body, by attaching together pieces of type 0, 1 and 2. Each piece is itself a handle-body which contains $\lesssim 1$ handles. There is a piece of type $i \in \{0,1,2\}$ in $M_h$ corresponding to each basis vector of $C_i$ and each piece has volume $\lesssim 1$.These pieces are glued together with certain attaching maps given in \cite{FH}. To make a good cover we would like to take a small neighborhood of each piece. However, to ensure that this cover is good we need to perturb the attaching maps so that two pieces of the same time are distance at least $D$ from each other, for some small $D \gtrsim 1$.

\par Each a piece of type 0 is a trivial bundle over a $(k-1)$-sphere $S^{k-1}$ and each piece of type 1 is a trivial bundle over an $S^k$ punctured in $\lesssim 1$ points. Every piece of type 1 attaches to several pieces of type 0 and at each piece of type 0, its attached in a small neighborhood of a $(k-1)$-sphere. We denote the union of the pieces of type 0 by $M_h(0)$, and denote the union of pieces of type 0 with pieces of type 1 attached by $M_h(1)$. After a finite subdivision the distance between the punctures in a piece of type 1 is at least $D$ for some $D \gtrsim 1$. This means that that the distance between two pieces of type 0 in $M_h(1)$ is at least $D$.

 Next, we want to make sure pieces of type 1 are not too close in $M_h(1)$. To ensure this we will apply the following quantitative general position argument to the attaching maps. Its proof is similar to the argument in \cite{GG} and it is also a warm-up to the proof in Appendix 1 of a more general quantitative embedding theorem.

\begin{lemma} \label{gplemma} Let $\tilde M$ be a triangulated $\tilde d$-manifold with bounded geometry. For any dimension $\tilde k \le (\tilde d-1)/2$, let $Y$ be a (not necessarily connected) $\tilde k$-dimensional simplicial complex with bounded geometry. Also let $G_1: Y \to \tilde M$ by any map which maps each simplex of $Y$ into a ball of radius $1/100$ in $\tilde M$, and $G_1^{-1}$ of any unit ball in $\tilde M$ intersects $\lesssim 1$ simplices in $Y$. Then there is an embedding $G_2: Y \to \tilde M$, which is homotopic to $G_1$, and for any two simplices $\sigma_1, \sigma_2$ in $Y$ which don't share a vertex we have, $dist(G_2(\sigma_1), G_2(\sigma_2)) \ge D$, for some small $D \gtrsim 1$. \end{lemma}

\begin{proof} We will find $G_2$ using a random perturbation of $G_1$. For each vertex $v$ in $Y$, let $G_2(v)$ be a uniformly random point in a ball of radius $1/100$ around $G_1(v)$. By the metric we endow $\tilde M$ with, each unit ball in $\tilde M$ is $\lesssim 1$-bilipschitz to a unit ball in $\mathbb{R}^{\tilde d}$. Since $G_2$ maps all the vertices of a simplex in $Y$ into a unit ball and each unit ball in $\tilde M$ is a convex set, we can extend $G_2$ linearly to all of $Y$. We can use a linear homotopy to show that $G_2$ is homotopic to $G_1$. Now, we show that $G_2$ satisfies the conditions of the theorem with non-zero probability.

\par For two simplices $\sigma_1, \sigma_2$ in $Y$, which don't share a vertex, and for a $D < 1$, let $Bad_D(\sigma_1, \sigma_2)$ denote the event that $dist(G_2(\sigma_1), G_2(\sigma_2)) < D$. Then the probability of $Bad_D(\sigma_1, \sigma_2)$ is bounded by the probability that $G_2(\sigma_1)$ and $G_2(\sigma_2)$ both intersect some ball of radius $D$. Let $B$ be some fixed ball of radius $D$ in $\tilde M$. The probability that $G_2(\sigma_1)$ intersects $B$ is $\lesssim D^{\tilde d- \tilde k}$. Note that $\sigma_1$ and $\sigma_2$ are disjoint, so their images under $G_2$ are independent. Thus, the probability that $G_2(\sigma_1)$ and $G_2(\sigma_2)$ both intersect $B$ is $\lesssim D^{2\tilde d- 2\tilde k}$. If they both intersect some ball, then this ball of radius $D$ must intersect a $(D+1/100)$-neighborhood of $G_1(\sigma_1)$. So using a union bound, we can bound the probability that $\sigma_1$ and $\sigma_2$ both intersect some ball of radius $D$ by $\lesssim D^{-\tilde d} D^{2\tilde d- 2\tilde k} \le D$. Thus, $Prob(Bad_D(\sigma_1, \sigma_2)) \lesssim D$.

\par Each event $Bad_D(\sigma_1, \sigma_2)$ is dependent on $\lesssim 1$ other such events, since $G_1^{-1}$ of a unit ball intersects $\lesssim 1$ simplices of $Y$. Thus, for some small $D \gtrsim 1$ we can apply the Lovasz-Local lemma \cite{EL}, to conclude that none of the events $Bad_D(\sigma_1, \sigma_2)$ occur with non-zero proability. This shows that we can pick a $G_2$ which satisfies the conditions in the theorem.
\end{proof}

After a finite subdivsion, we can ensure that the attaching maps map each simplex in the $(k-1)$-sphere of a piece of type 1 into a ball of radius $1/100$ in a piece of type 0. The \cite{FH} construction also says that $\lesssim 1$ pieces of type 1 attach to any fixed piece of type 0. Using these observaions we can apply \cref{gplemma} with the following inputs: $\tilde M = \partial M_h(0)$, $\tilde d = d-1$, $\tilde k = k-1$, $Y$ is all the $S^{k-1}$'s where pieces of type 1 attach to $\tilde M$, and $G_1$ is the attaching map of these spheres to $\tilde M$. \cref{gplemma} says that we can perturb the attaching map $G_1$ so that any two pieces of type 1 are distance $\ge D$ from each other, for some small $D \gtrsim 1$.  

\par Now we consider the attaching maps for pieces of type 2. Each piece of type 2 is a trivial bundle over a manifold $N_{\Gamma}$, where $N_{\Gamma}$ is a sphere $S^{k+1}$ with a neighborhood of a graph $\Gamma \subset S^{k+1}$ removed. The relevant property of $\Gamma$ for us is that it has bounded geometry and volume $\lesssim 1$. A piece of type 2 attaches to several pieces of type 0 and type 1, and the attaching happens in a small neighborhood of $\partial N_{\Gamma}$, which is $k$ dimensional. Attaching pieces of type 2 can potentially decrease the distance between pieces of type 1 or 0 in $M_h$ by a significant amount. After examining the attaching maps of pieces of type 2 in \cite{FH}, we can see that this doesn't happen if the following condition holds. Any two vertices of $\Gamma$ are distance $\ge D$ in $S^{k+1}$, and any two edges of $\Gamma$ that don't share a vertex are distance $\ge D$ in $S^{k+1}$. We can ensure this condition by perturbing the inclusion of $\Gamma$ in $S^{k+1}$. Apply \cref{gplemma} with $\tilde M = S^{k+1}$, $Y$ is a finite subdivision of $\Gamma$, and $G_1$ is the inclusion of $\Gamma$ into $S^{k+1}$. 

\par Next, we need to ensure that pieces of type 2 are not too close together. From \cite{FH},we have that $\lesssim 1$ pieces of type 2 attach to any fixed piece of type 0 or 1. Then we can apply \cref{gplemma} with the following inputs: $\tilde M = \partial M_h(1)$, $\tilde d = d-1$, $\tilde k = k$, $Y$ is several copies of $\partial N_{\Gamma}$ where all the pieces of type 2 attach to $\tilde M$, and $G_1$ is the attaching map of these $\partial N_{\Gamma}$ to $\tilde M$. \cref{gplemma} says that we can perturb the attaching map $G_1$ so that any two pieces of type 2 are distance $D$ from each other, for some small $D \gtrsim 1$. 

\par The above discussion shows that the $D$-neighborhoods of all the pieces form a good cover of $M$ of multiplicity 3. This completes the proof of the lemma.
\end{proof}

Let $X$ be the nerve of the good cover from \cref{multlemma} and denote the nerve map by,

$$\rho: M_h \to X$$

\noindent From the observations at the end of the section on nerves we have that $X$ is a 2-dimensional simplicial complex of bounded degree with volume $\lesssim V$. Our goal now is to subdivide $X$ and map this subdivision into $\mathbb{R}^n$. We do this using the following theorem which is a slight modification of the result from \cite{GG}. The proof can be found in Appendix 1. 

\begin{theorem} Let $X$ be an $m$-dimensional simplicial complex of degree $A \lesssim 1$ and volume $V$.  Fix a dimension $n > m$ and let $R = V^{\frac{1}{n-m}} \, log^{n+1}(V)$. For $V$ large enough, we can construct a subdivision of $X$ of bounded degree, which we call $X'$, and a coarse map, 

$$I: X' \to B^n(R)$$

\noindent where $B^n(R)$ is a ball of radius $R$ in $\mathbb{R}^{n}$. $X'$ is obtained by subdividing each $m'$-simplex of $X$ into $\approx R^{m'}$ standard simplices for each $1 \le m' \le m$. Furthermore, if $\sigma$ is an $m'$-simplex in $X'$ for some $1\le m' \le m$, then $I(\sigma)$ is $\lesssim 1$-bilipschitz to a standard unit $m'$-simplex.
\end{theorem}

We can now apply this theorem to our 2-complex $X$. This gives us a subdivision of $X$, which we'll call $X'$, and a coarse map $I: X' \to B^n(R)$ with $R \approx_V V^{\frac{1}{n-2}}$. However, we still need to subdivide to $M$ to construct a coarse map from it to $X'$. We do this using the following lemma.

\begin{lemma} \label{sublemma} Fix dimensions $d \ge m \ge 0$. Let $M$ be a triangulated $d$-manifold with bounded degree, $X$ a $m$-complex with bounded degree, and $\tilde F: M \to X^m$ a coarse map which is $L$-Lipschitz for some $L \lesssim 1$. Let $X'$ be a simplicial complex of bounded degree obtained by subdividing each $m'$-simplex of $X$, for $1 \le m' \le m$, into $\approx A^{m'}$ simplices for some number $A \ge 1$. Then we can make a subdivision of $M$, denoted $M'$, with $vol(M') \lesssim vol(M)A^{m}$ and construct a coarse map $F': M' \to X'$.
\end{lemma}

\begin{proof} First barycentrically subdivide $M$ several times, so that each new simplex has side-length roughly $(100^dL)^{-1}$, and call the resulting subdivided manifold $M_s$. Because $\tilde F$ is $L$-Lipschitz on $M$, we can use simplicial approximation to homotope $F$ to a simplicial map $F': M_s \to X$. Note that $dist_X(F(p), F'(p)) \lesssim 1$ for any $p \in M$. Sees Lemma 1.5 in \cite{GQ} for more details on this homotopy. Now $F'$ maps each simplex of $M_s$ linearly to some simplex of $X$. Since $\tilde F$ was coarse, $F'$ is coarse as well.
\par Since $F'$ is linear on each simplex, we can use $F'$ to pull-back the simplicial structure on $X'$ to a polyhedral subdivision of $M_s$, which we denote by $M'_s$. For instance, suppose a simplex $\tau \in M_s$ gets mapped to a simplex $\sigma \in X$ and $\sigma' \subset \sigma$ is a simplex in $X'$. Then $\tau \cap F^{-1}(\sigma')$ is a polyhedron in $M'_s$. 

\par Each polyhedron lies inside a simplex of $M_s$ and has $\lesssim 1$ vertices. Since $M$ and $X'$ have bounded degree, we can subdivide each polyhedron of $M'_s$ into $\lesssim 1$ simplices, to get a triangulation of bounded degree. Let $M'$ be this triangulated manifold. Since we obtained $X'$ by subdividing each simplex of $X$ into $A^{m}$ simplices and each polyhedron in $M'_s$ was subdivided $\lesssim 1$ times, we have $vol(M') \lesssim vol(M)A^{m}$.
\end{proof} 

Recall that $\rho$ is $\lesssim 1$-Lipschitz, since it is a nerve map of a good cover, and that $X'$ is obtained by subdividng each $m'$-simplex of $X$ into $\approx R^{m'}$ simplices. This means $\tilde F = \rho \circ F_h$  and $X'$ satisfy the hypothesis of \cref{sublemma}. Applying \cref{sublemma} to this $\tilde F$ and $X'$, we get a subdivision of $M$, denoted $M'$, so that $vol(M') \lesssim R^2 vol(M)$ as well as a coarse map $F': M' \to X'$. Finally, we have a map $F = I \circ F': M' \to B^n(R)$, which is coarse by \cref{complemma}.

\par Next we'll analyze the behavior of the systoles in $M'$. Since each piece used to construct $M_h$ is a handlebody and $M$ is the double of $M_h$, we can identify a set of spines of these handlebodies in $M$. Lemma 1.6.2 in \cite{FH} shows that any singular Lipschitz $k$ or $d-k$ cycle in $M$ is homologous to a cellular cycle supported on these spines, which we'll denote by $spine(z)$. Lemma 1.6.2 is proved by deforming the part of $z$ inside a handle to that handle's spine. So if $z$ did not intersect some piece of $M$, then $spine(z)$ does not intersect that piece either. Now the relationship of $M$ to the code it is associated to says that if $spine(z)$ is not a boundary, then it must contain at least $d(\mathcal{C'})$ spines in its support. In particular, if we let $z'$ be a simplicial $k$ or $d-k$ cycle of $M'$ that is not a boundary, then $z'$ must intersect at least $d(\mathcal{C'})$  pieces of $M$. In fact, since each piece of $M_h$ intersects $\lesssim 1$ other pieces and since each piece consists of $\lesssim 1$ simplices, we can ensure that $z'$ intersects $\gtrsim d(\mathcal{C'})$ disjoint pieces of $M$.

\par Lets now consider $F'(z')$ as a subset in $X$ (rather than $X')$. From the construction of the nerve map $\rho$, we see that $\tilde F(z')$ intersects $\gtrsim d(\mathcal{C'})$ simplices in $X$. And since $dist_X(F'(p),\tilde F(p)) \lesssim 1$ for any $p \in M'$, we also see that $F'(z')$ intersects $\gtrsim d(\mathcal{C'})$ simplices in $X$. If we assume that $z'$ is connected, then by the following lemma we have that $F'(z')$ must intersect $\gtrsim d(\mathcal{C'})R$ simplices of $X'$. 

\begin{lemma} Let $Y$ be a connected sub-complex of $X'$ that intersects at least $N$ simplices of $X$. Then $Y$ contatins $\gtrsim NR$ 1-simplices of $X'$.
\end{lemma}

\begin{proof} It is enough to consider the case when $Y$ is a 1-complex. We have a metric on $X'$ where each simplex has side-length 1.  Note that two vertices of $X$ are roughly distance $R$ apart under this metric. Consider a cover of $Y$ by balls of radius $R/100$ centered on $Y$, so that the center of each ball does not lie in any other ball. As $X$ has bounded degree, each ball intersects $\lesssim 1$ simplices of $X$. Thus, there must at least $N' \gtrsim N$ balls in the cover. Since $Y$ is connected and the centers of the balls are $\gtrsim R/100$ distance from each other, we have $|Y| \gtrsim N' R/200 \gtrsim NR$.
\end{proof}

Since $F'$ is a coarse map we have $vol(z') \gtrsim vol(F'(z')) \gtrsim d(\mathcal{C'})R$. Since $z'$ was any $k$ or $d-k$ cycle in $M'$ that was not a boundary, we have shown that $sys_k(M')$ and $sys_{d-k}(M')$ are $\gtrsim d(\mathcal{C'})R$. Putting everything together we have,
$$vol(M') \lesssim_V V^{1+\frac{2}{n-2}}$$
$$dim(H_k(M')) = dim(\mathcal{C'})$$
$$sys_k(M') \gtrsim_V d(\mathcal{C'})V^{\frac{1}{n-2}}$$
$$sys_{d-k}(M') \gtrsim_V d(\mathcal{C'})V^{\frac{1}{n-2}}$$

\par Although $F$ is coarse, $M'$ is not local because $F$ maps into a ball of volume roughly $R^n = V^{\frac{n}{n-2}} log(V)^{n(n+1)}$, while $vol(M')$ is roughly $V^{\frac{n}{n-2}} log(V)^{2(n+1)}$. Apply \cref{enlargelemma} to get a local manifold $M''$ from $M'$. $M''$ has the same homology as $M'$ and roughly the same systoles as $M'$ in dimension $k$ and $d-k$, and so satisfies the conditions stated in the theorem.
\end{proof}
 
We can now use this result to prove \cref{Pthm}.

\begin{proof}[Proof of \cref{Pthm}] Let $M$ be the local manifold constructed in \cref{mfldlemma} with $vol(M) \lesssim_V V^{\frac{n}{n-2}}$ and $\min(sys_k(M), sys_{d-k}(M)) \gtrsim_V d(\mathcal{C'})V^{1/n}$.  In order to apply \cref{MClemma} we need an estimate on the cosystole of $M$ for which we use the following lemma.

\begin{lemma} \label{cosyslemma} Let $M$ be a $d$-dimensional triangulated manifold with bounded degree. Then for $0 < k < d$ we have,
$$sys^{k}(M) \gtrsim sys_{d-k}(M)$$
\end{lemma}

\begin{proof} The follwowing argument is similar to a proof of Poincare Duality which relies on the dual polyhedral structure of a triangulation. Some more details can be found in Chapter 5.2 of \cite{S}. Let $\tau$ be the given triangulation of $M$. By our definition of a triangulated manifold, the simplices adjacent to a vertex form a set that is homeomorphic to $\mathbb{R}^d$. Let $\tau'$ be the subdivision of $M$ into polyhedra which is dual to $\tau$. That is, $\tau'$ has a vertex in every $d$-simplex of $\tau$, an edge that transversely crosses each $(d-1)$-simplex of $\tau$, and so on. To define $\tau'$ more formally suppose we have a $k$-simplex $\sigma \in \tau$ and $d$-simplex $\Delta \in \tau$ such that $\sigma \subset \Delta$. The part of the polyhedron dual to $\sigma$ inside $\Delta$ is the convex hull of the barycenters of all simplicies in $\Delta$ that contain $\sigma$.

\par For each $k$-cochain $z$ in $\tau$, define $D(z)$ be the set of all $(d-k)$-polyhedra in $\tau'$ which intersect the simplices in the support of $z$. By the construction of $\tau'$, we have $vol(D(z)) = vol(z)$, where volume in $\tau'$ is measured by counting the number of polyhedron in the support. Notice that a $(d-k-1)$-polyderon of $\tau'$ is in $\partial D(z)$ if and only it intersects a $(k+1)$-simplex in $\delta^{k}(z)$. Thus, $\partial D(z) = D(\delta^{k}(z))$. So we see that $D$ is a chain map from the $k$-cochains in $\tau$ to $(d-k)$-chains in $\tau'$.

\par We now show how to deform any $(d-k)$-chain $w \in \tau'$, to a $(d-k)$-chain in $\tau$ using a projection argument. Suppose $\sigma$ is an $m$-simplex in $\tau$, with $m>d-k$ or $\sigma$ is a $(d-k)$-simplex in $\tau$ which is not entirely contained in $w$. Then there is some point $p \in \sigma^{\circ} - w$. Let $\pi_p: (\sigma - p)  \to \partial \sigma$ be the projection map in 
$\sigma$ centered at $p$. Then the following chain is homologous to $w$ and does not intersect $\sigma^{\circ}$,

$$w - (w \cap \sigma) + \pi_p(w \cap \sigma)$$

Applying this procedure several times we eventually end up with a chain in $\tau$. Observe that since the triangulation $\tau$ has bounded degree, this procedure increases the volume of the original chain by at most a constant factor. Now use this procedure, to deform $D(z)$ to a chain $P(z)$ in $\tau$. Because of the observation we have  $vol(P(z)) \lesssim vol(D(z)) = vol(z)$.

\par Let $z$ be a cocycle that is not a coboundary with $vol(z) = sys^k(M)$. Because $D$ is a chain map we have that $P(z)$ is a cycle that is not a boundary, with $vol(P(z)) \lesssim sys^k(M)$. This gives the desired bound on $sys_{d-k}(M)$. 
\end{proof}

By \cref{cosyslemma} we have that $sys^k(M) \gtrsim_V d(\mathcal{C'})V^{1/n}$. So using \cref{locallemma} we see that there is a local code $\mathcal{C}$, with $size(\mathcal{C}) \lesssim_V V^{\frac{n}{n-2}}$, $dim(\mathcal{C}) = dim(\mathcal{C'})$, and $d(\mathcal{C}) \gtrsim_V d(\mathcal{C'})V^{\frac{1}{n}}$. This completes the construction of the local code.
\end{proof}

\noindent We can now use this result to prove \cref{Cthm}.

\begin{proof}[Proof of \cref{Cthm}] Let $\mathcal{C'}$ be an asymptotically good LDPC code that admits a sparse lift, for instance the code from \cite{PK} or \cite{LZ}. Apply \cref{Pthm} to $\mathcal{C'}$ to get a local code $\mathcal{C}$. The properties of $\mathcal{C'}$ from \cref{Pthm} give the desired lower bounds on $dim(\mathcal{C})$ and $d(\mathcal{C})$. The upper bounds follow from the results in Appendix 2.
\end{proof}

\section*{\centering Appendix 1: Quantitive embedding result}

In this appendix we give a slight variation on the quantitive embedding theorem of Gromov-Guth that we use in the proof of \cref{mfldlemma}. The proof ideas are all from Theorem 2.1 in \cite{GG}.

\begin{theorem*} [7] Let $X$ be an $m$-dimensional simplicial complex of degree $A \lesssim 1$ and volume $V$.  Fix a dimension $n > m$ and let $R = V^{\frac{1}{n-m}} \, log^{n+1}(V)$. For $V$ large enough, we can construct a subdivision of $X$ of bounded degree, which we call $X'$, and a coarse map, 

$$I: X' \to B^n(R)$$

\noindent where $B^n(R)$ is a ball of radius $R$ in $\mathbb{R}^{n}$. $X'$ is obtained by subdividing each $m'$-simplex of $X$ into $\approx R^{m'}$ standard simplices for each $1 \le m' \le m$. Furthermore, if $\sigma$ is an $m'$-simplex in $X'$ for some $1\le m' \le m$, then $I(\sigma)$ is $\lesssim 1$-bilipschitz to a standard unit $m'$-simplex.
\end{theorem*}

\begin{proof} We begin by mapping $X$ into an $n$-ball of radius $V^{\frac{1}{n-m}}$. Choose a coloring of the vertices of $X$ with $A+1$ colors, so that no two adjacent vertices are the same color.
Next, we select $A+1$ disjoint caps in $\partial B^n(V^{\frac{1}{n-m}})$. We can choose them so that each cap has volume $\varepsilon V^{\frac{m}{n-m}}(A+1)^{-1}$ for some small constant $\varepsilon \gtrsim 1$, and so that any two caps are $\gtrsim V^{\frac{1}{n-m}}$ distance apart. We now define a random map $I_0:X \to B^n(V^{\frac{1}{n-m}})$. That is, the image of each simplex under $I_0$ will be a random variable, as follows. $I_0$ maps each vertex of color $i \in \{1, 2 \ldots A+1\}$ to a uniformly random point in the $i^{th}$ cap. Then extend $I_0$ by linearity to all other simplices of $X$. Because the caps are roughly $V^{\frac{1}{n-m}}$ distance apart, the image of each simplex is $\lesssim 1$-bilipschitz to a standard simplex of side-length $V^{\frac{1}{n-m}}$.
\par
Now, color all the $m$-simplices with $\lesssim A$ colors, so that no two simplices that share a vertex are the same color.
Note that the image of the simplices of the same color can be treated like independent events, as they share no vertices. The probability that an $I_0$-image of a simplex in $X$ intersects some fixed unit ball, is roughly the probability that a fixed simplex intersects a uniformly random ball in $B^n(V^{\frac{1}{n-m}})$. This probability is roughly,

$$\frac{vol(m\text{-simplex of side-length  } V^{\frac{1}{n-m}})}{vol(B^n(V^{\frac{1}{n-m}}))} = V^{-1}$$

Thus, the expected number of $m$-simplices of the same color that intersect some unit ball $B^n(1)$ is $\lesssim 1$. A Chernoff bound then shows that the probability that $\gtrsim log(V)$ simplices of the same color intersect some $B^n(1)$ is at most $e^{-C'log(V)}$, for some small constant $C'>0$. Therefore, with non-zero probability there is a map $I_1: X \to B^n(V^{\frac{1}{n-m}})$, so the the pre-image of any unit ball intersects $\lesssim log(V)$ simplices. 

\par Let $s = \delta\, log^{n+1}(V)$ be a scaling factor, where $\delta > 0$ is some tiny universal constant. Scale the image of $I_1$ by a factor of $s$, to get a map $I_2: X \to B^n(sV^{\frac{1}{n-m}})$. Note that $I_2^{-1}$ of any ball of radius $s$ intersects $\lesssim log(V)$ simplices of $X$. We can now subdivide each $m'$-simplex of $X$ into roughly $V^{\frac{m'}{n-m}}$ simplices, so that $I_2$ maps each new simplex into a ball of radius $s$. Call the subdivided complex we obtain $X'$. The pre-image of any ball of radius $s$ still intersects $\lesssim log(V)$ simplices of $X'$. 

\par The next step is to randomly perturb the map $I_2$ at scale $s$. Color the vertices of $X'$ in $A' \lesssim 1$ colors, so that no two intersecting simplices of $X'$ contain 2 vertices of the same color. Also select $A'$ caps in a sphere of radius $s$ which are $\gtrsim s$ distance apart from each other. For each vertex $v \in X'$ of color $i$, let $y(v)$ be a random vector in the $i^{th}$ cap. Define $I_3(v) = I_2(v)+y(v)$, and extend linearly to all the simplices of $X'$ to get a map $I_3: X' \to B^n(sV^{\frac{1}{n-m}})$. The coloring ensures that after this perturbation is done, the $I_3$-image of each simplex is $\lesssim 1$-bilipschitz to a standard simplex with side-length $s$.
\par
 $I_3$ is still defined as a random map, and our goal is to show that with high probability it is coarse. Color all the simplices with $A'' \lesssim 1$ colors, so that no two simplices that share a vertex are the same color. For a set of $(n+1)$ $m$-simplices of $X'$ with the same color, denoted $\{\Delta_i\}_{i=1}^{n+1}$, let $Bad(\{\Delta_i\})$ to be the probability that there is some unit ball that intersects all $n+1$ of these simplices. Notice that since the perturbation was done at scale $s$, all the $\Delta_j$ have to be inside a ball of radius $\lesssim s$ for $Bad(\{\Delta_i\})$ to be non-zero. The coloring also ensures that the perturbations of vertices in any two simplices of the same color are independent events. Suppose we fix a unit ball inside some ball of radius $s$, which we can denote by, $B^n(1) \subset B^n(s) \subset B^n(R)$. Then the probability that a random $m$-simplex of side-length $s$ in $B^n(s)$ intersects $B^n(1)$ is about $s^{m-n}$. Thus we can estimate,

$$Bad(\{\Delta_i\}) \lesssim s^n (s^{m-n})^{(n+1)} \le s^{-1}$$

\noindent Where the last inequality requires $n>m$. Also notice that each bad event $Bad$ is dependent on $\lesssim log^{n+1}(V)$ other bad events, because of the pre-image condition on $I_2$. Now recall that $s= \delta\, log(V)^{n+1}$. So for small enough $\delta \gtrsim 1$ we have,

$$log^{n+1}(V) Bad(\{\Delta_i\}) < 1/100$$
 
 By the Lovasz-Local lemma \cite{EL}, there is some non-zero chance that none of the bad events $Bad(\{\Delta_i\})$ occur. In other words, there is a non-zero chance that $I_3^{-1}$ of any unit ball intersects at most $(n+1)A''$ simplices. Let $I$ be the map with this property. Now subdivide each $m'$-simplex of $X'$ into roughly $s^{m'}$ simplices so that $I$ maps each new simplex into a ball radius 1. This completes the construction of $I$ and $X'$.
\end{proof}

\section*{\centering Appendix 2: Upper bounds on local codes}

In this section we given some topological proofs of upper bounds on the distance and dimension of local codes. The results were previously proved in \cite{B}, \cite{BT}, \cite{BPT}, and \cite{FHK}.

\begin{theorem} \label{lcbound} For $n \ge 1$, if $\mathcal{C}$ is a local code of size $V$ then, 
$$d(\mathcal{C}) \lesssim V^{(n-1)/n}$$
\end{theorem}

\begin{theorem} \label{pdbound} For $n \ge 2$, if $\mathcal{C}$ is a local code of size $V$ then, 
$$dim(\mathcal{C})d(\mathcal{C})^{2/(n-1)} \lesssim V$$
\end{theorem}

\noindent If a local code of size $V$ satisfies $d(\mathcal{C}) \gtrsim_V V^{(n-1)/n}$, then \cref{pdbound} shows that $dim(\mathcal{C}) \lesssim_V V^{(n-2)/n}$. So, in particular, the code we construct in \cref{Cthm} has the best possible dimension, up to a $polylog(V)$ factor, given its distance.

\par We will prove these theorems by transforming them into statements about manifolds. The next lemma tells us we can convert a local code into a local manifold.

\begin{lemma} If $\mathcal{C}$ is a local code of size $V$, then the associated manifold from \cref{CMlemma} is local.
\end{lemma}

\begin{proof} Let $\mathcal{B}$ be a basis for $C_1$, which is the middle vector space in the chain complex of $\mathcal{C}$. Let $f: \mathcal{B} \to Q^n(V)$ be the map given from the definition of a local code. Let $\tilde F = \rho \circ F_h: M \to X$ be the map from the proof of \cref{Cthm}. For each vertex $v \in X$ pick a piece of $M$ that intersects $(\tilde F)^{-1}(v)$. This piece is distance $\lesssim 1$ from some piece of type 0. This piece of type 0 corresponds to an vector in $\mathcal{B}$ which we'll label $b_v$. Viewing $Q^n(V)$ as a subset of $\mathbb{R}^n$ we can then construct a map $\tilde I: X \to \mathbb{R}^n$ by letteing $\tilde I(v) = f(b_v)$ and extending by linearity. The properties of $f$ then imply that $F = \tilde I \circ \tilde F: M \to \mathbb{R}^n$ is coarse.
\end{proof}

Using this lemma, \cref{lcbound} and \cref{pdbound} now follow from the following theorems.

\begin{theorem} \label{lcboundm} Fix dimensions $n \ge 1$, $d \ge 2$ and $0 < k < d$. If $M$ is a local $d$-dimensional manifold with $vol(M) = V$ then, 
$$\min(sys_k(M), sys_{d-k}(M)) \lesssim V^{(n-1)/n}$$
\end{theorem}

\begin{theorem} \label{pdboundm} Fix dimensions $n \ge 2$, $d \ge 2$ and $0 < k < d$. If $M$ is a local $d$-dimensional manifold with $vol(M) = V$ then, 
$$dim(H_k(M)) \min(sys_k(M), sys_{d-k}(M))^{2/(n-1)} \lesssim V$$
\end{theorem}

For the proof of these theorems we'll use the following notation. For a set $S \subset \mathbb{R}^n$ and a map $F: M \to \mathbb{R}^n$, let $F_T^{-1}(S)$ denote a tiny open neighborhood of the set of simplicies that intersect $F^{-1}(S)$.

\begin{proof}[Proof of \cref{lcboundm}] This proof was communicated by Larry Guth and is a topological version of the argument in \cite{B}. By assumption, we have a coarse map $F: M \to B$, where $B$ is a ball of volume $V$ in $\mathbb{R}^n$. This ball can be covered by two disjoint sets of slabs, $U_1$ and $U_2$, so that most slabs roughly looks like a 2-neighborhood of a $(n-1)$-dimension ball of volume $V^{(n-1)/n}$. More precisely, 
$$U_1 = \{(x_1, \ldots x_n) \in  B:  x_1\textbf{ mod }3 \in [0,2]\}$$
$$U_2 = \{(x_1, \ldots x_n) \in  B:  x_1 \textbf{ mod } 3 \in [1,3]\}$$

\par Let $\tilde U_1 = F_T^{-1}(U_1)$ and $\tilde U_2 = F_T^{-1}(U_2)$ in $M$. Now we apply a Lyusternik-Schnirelmann argument to show that $\tilde U_1$ or $\tilde U_2$ must support a non-trivial $k$ or $d-k$ cocycle. Suppose for contradiction, $H^k(M)$ and $H^{d-k}(M)$ have no classes supported on either $\tilde U_1$ or $\tilde U_2$. Then $H^k(M) = H^k(M, \tilde U_1)$ and $H^{d-k}(M) = H^{d-k}(M, \tilde U_2)$. Since $\tilde U_1 \cup \tilde U_2 = M$, we have a trivial cup product pairing $H^{k}(M, \tilde U_1) \cup H^{n-k}(M, \tilde U_2) = 0$. This is a contradiction, since $H^{k}(M) \cup H^{d-k}(M) \neq 0$ by Poincare duality. We can again apply Poincare duality to see that one of the following two inclusion maps has non-zero image.

$$H_k(\tilde U_1) \to H_k(M)$$
$$H_k(\tilde U_2) \to H_k(M)$$

Let's suppose the first map is non-zero. The case when the second map is non-zero is analogous. Let $z$ be some $k$-cycle which not a boundary and supported on $\tilde U_1$. Then some connected component of $z$, call it $z'$, is also not a boundary. Note that $z'$ must lie in the pre-image of some slab. By using a projection in each simplex, as in the proof of \cref{cosyslemma}, we see that, $z'$ is homologous to a set of $k$-simplices in the pre-image of some slab, denoted by $z''$. Since $F$ is a coarse embedding and each slab has volume roughly $vol(M))^{(n-1)/n}$, we have,
$$sys_k(M) \le vol(z'') \lesssim vol(M))^{(n-1)/n}$$
Switching $k$ and $d-k$ in the argument above, we get the same bound on $sys_{d-k}(M)$. This proves the bound on $\min(sys_k(M), sys_{d-k}(M))$.
\end{proof}

\begin{proof} [Proof of \cref{pdboundm}]  This is a topological version of the argument in \cite{FHK}. The idea of the proof is to show that all the $k$-dimensional homology classes of the manifold are supported on the pre-image of a large $(n-1)$-dimensional lattice. This is done by using Poincare Duality to deform $k$-cycles to a lattice. Then we use the Mayer-Vietoris sequence and simplicial homology to bound the number of such classes that can be supported on this lattice.

\par By assumption there is a coarse map $F: M \to B$ where $B$ is a ball of volume $vol(M)$ in $\mathbb{R}^n$. Define $sys =  \min\{sys_k(M), sys_{d-k}(M)\}$. Let $s_M > 0$ be the largest number so that the following holds. If $Q_{s_M}$ is any $n$-cube in $\mathbb{R}^n$ of sidelength $s_M$, the following inclusion maps are 0,

$$i_k: H_k(F_T^{-1}(Q_{s_M})) \to H_k(M)$$ 
$$i_{n-k}: H_{d-k}(F_T^{-1}(Q_{s_M})) \to H_{d-k}(M)$$

Our goal will be to show that $s_M^{n-1}$ cannot be much smaller than $sys$. The main tool is the following deformation lemma,

\begin{lemma} Let $N$ be an open subset of a $d$-manifold $M$. Suppose for some $0<k<d$, we have that the inclusion map $i_{d-k}: H_{n-k}(N) \to H_{d-k}(M)$ is 0. Then any homology class in $H_k(M)$ has a representative which is supported on $M-N$.
\end{lemma} 

\begin{proof} Let $z$ be a representative of $H_k(M)$, and set $z_N = z \cap N$. If $z_N$ is trivial in $H_k(N, \partial N)$, then it is homologous to a $k$-chain $z_{\partial N}$ in $\partial N$, and so $z$ is homologous to $(z-z_N) + z_{\partial N}$, proving the theorem. So suppose $z_N$ is non-trivial in $H_k(N, \partial N)$. By Poincare Duality there exists some $w \in H_{d-k}(N)$, so that the intersection number of $z_N$ with $w$ is not 0. If $i: N \to M$ is the inclusion map, then the intersection number of $z$ with $i(w)$ would also not be 0. This implies that $i(w)$ is a non-trivial representative of a class in the image of $i_{d-k}$, which contradicts the assumption in the lemma.\end{proof}

Denote by $G^k_s$, the $k$-dimensional grid on $\mathbb{R}^n$ at scale $s>0$. That is the set of all points $(x_1, x_2 \ldots x_n) \in \mathbb{R}^n$ so that at least $(n-k)$ of the $x_i$ are integer multiples of $s$. Let $A' > 0$ by any positive constant and $Q_{A's_M}$ any cube of sidelength $A's_M$. Given these choices we consider the following set, 

$$N = F_T^{-1}(Q_{A's_M}) - \overline{F_T^{-1}(G^{n-1}_{s_M})}$$

\noindent Each connected component of $M-N$ lies in the $F^{-1}_T$ pre-image of a cube of side-length $s_M$ so the inclusion maps, $i_k$ and $i_{n-k}$ are 0 for this $N$. If $z$ was some $k$-cycle or $(d-k)$-cycle supported in $F_T^{-1}(Q_{A's_M})$, then by the deformation lemma it would be supported in $\partial N$. But by the coarseness of the map we have, 
$$vol(\partial N) \lesssim (A')^n s_M^{n-1}$$
So if $sys \lesssim (A')^n s_M^{n-1}$ then $z$ would be homologically trivial. Since the choice for $z$ and $Q_{A's_M}$ was arbitrary, the inclusion maps would be 0 for any cube of side-length $A's_M$. This would contradict the maximality of $s_M$, and so we see that $s_M \gtrsim sys^{1/(n-1)}$.
 \par
Using the definition of $s_M$ and the deformation lemma again, we see that all the $k$-dimensional homology classes can be supported on $F_T^{-1}(G^{n-1}_{\delta sys^{1/(n-1)}})$ for some small $\delta \gtrsim 1$. In other words, the following inclusion map is surjective,

$$j_k: H_k(F_T^{-1}(G^{n-1}_{\delta sys^{1/(n-1)}})) \to H_k(M)$$

\noindent Now consider the following two sets that cover $F_T^{-1}(G^{n-1}_{\delta sys^{1/(n-1)}})$,

$$N_1 = F_T^{-1}(G^{n-2}_{\delta sys^{1/(n-1)}})$$
$$N_2 = F_T^{-1}(G^{n-1}_{\delta sys^{1/(n-1)}}) - F_T^{-1}(G^{n-2}_{\delta sys^{1/(n-1)}})$$

\noindent $N_2$ is roughly the pre-image of several disconnected $(n-1)$-dimensional cubes, each of volume $\lesssim \delta^{n-1} sys$. If $\delta$ is small enough we see that no nontrivial $k$-dimensional homology of $M$ can be supported on $N_2$. 
\noindent From the Mayer-Vietoris sequence of $N_1$ and $N_2$ we see that,

$$dim(H_k(M)) = dim(im(j_k)) \le dim(H_{k}(N_1)) + dim(H_{k-1}(N_1))$$

\noindent Observe that since we are using simplicial homology, if a simplicial complex $X$ has bounded geometry, then have $dim(H_*(X)) \lesssim vol(X)$. Putting this together with the previous estimate we get,
$$dim(H_k(M)) \le dim(H_{k}(N_1)) + dim(H_{k-1}(N_1)) \lesssim vol(N_1) \lesssim vol(M) sys^{-2/(n-1)}$$
\end{proof}

\section*{\centering Declarations}

\noindent \textbf{Data availability.} Data sharing is not applicable to this article as no datasets were generated or analyzed during
the current study. 
\\ \\
\noindent \textbf{Conflict of interest.} The author declares no conflicts of interest.


\end{document}